\documentclass[runningheads,envcountsame,a4paper]{llncs}

\usepackage{amsmath,amssymb}

\DeclareMathOperator{\uf}{uf}
\DeclareMathOperator{\su}{su}
\DeclareMathOperator{\ufp}{ufp}
\DeclareMathOperator{\ufs}{ufs}

\frontmatter

\title{Factorization in Formal Languages}

\author{Paul Bell\inst{1} \and
Daniel Reidenbach\inst{1} \and
Jeffrey Shallit\inst{2}}

\institute{Department of Computer Science,
Loughborough University,
Loughborough,
Leicestershire,
LE11 3TU, United Kingdom\\
\email{P.Bell@lboro.ac.uk}\\
\email{D.Reidenbach@lboro.ac.uk} \\
\and
School of Computer Science,
University of Waterloo,
Waterloo, ON  N2L 3G1, Canada\\
\email{shallit@cs.uwaterloo.ca}}

\begin{document}

\maketitle

\begin{abstract}
We consider several novel aspects of unique factorization in formal languages.
We reprove the familiar fact
that the set $\uf(L)$ of words having unique factorization into elements
of $L$ is regular if $L$ is regular, and from this deduce an 
quadratic upper and lower bound on the
length of the shortest word not in $\uf(L)$.  We observe that $\uf(L)$
need not be context-free if $L$
is context-free.  

Next, we consider variations on unique factorization.
We define a notion of ``semi-unique'' factorization, where every
factorization has the same number of terms, and
show that, if $L$ is regular or even finite, the set of words having such a 
factorization need not be context-free.  Finally, we consider additional
variations, such as unique factorization ``up to permutation'' and
``up to subset''.
\end{abstract}

\section{Introduction}

Let $L$ be a formal language.  We say $x \in L^*$ has
{\it unique factorization} if whenever
$$x = y_1 y_2 \cdots y_m = z_1 z_2 \cdots z_n$$
for $y_1, y_2, \ldots, y_m, z_1, z_2,\ldots, z_n \in L$
then $m = n$ and $y_i = z_i$ for $1 \leq i \leq m$.
If every element of $L^*$ has unique factorization into elements of $L$, then
$L$ is called a {\it code}.

Although codes have been studied extensively (see, for example,
\cite{Berstel&Perrin&Reutenauer:2010}), in this paper we look
at some novel aspects of unique factorization.

\section{Unique factorizations}

Given $L$, we define $\uf(L)$ to be the set of all elements of $L^*$ having
unique factorization into elements of $L$.
We recall the following familiar fact:

\begin{proposition}
If $L$ is regular, then so is $\uf(L)$.
\label{prop1}
\end{proposition}

\begin{proof}
If $L$ contains the empty word $\epsilon$ then no elements of
$L^*$ have unique factorization, and so $\uf(L) = \emptyset$.  So,
without loss of generality we can assume $\epsilon \not\in L$.

To prove the result, we show that the 
relative complement $L^* - \uf(L)$ is regular.
Let $L$ be accepted by a DFA $M$.
On input $x \in L^*$, we build an NFA $M'$ to guess two different
factorizations of $x$ and verify they are different.
The machine $M'$ maintains the single state of the DFA $M$ for $L$ as it
scans the elements of $x$,
until $M'$ reaches a final state $q$.
At this point $M'$ moves, via an $\epsilon$-transition, to a new 
kind of state that records pairs.
Transitions on these ``doubled" states still follow $M$'s
transition function in both coordinates,
with the exception that if either state is in
$F$, we allow a ``reset" implicitly to $q_0$.  Each implicit return
to $q_0$ marks, in a factorization, the end of a term.
The final states of $M'$ are the ``doubled" states
with both elements in $F$.  

More precisely, assume $M = (Q, \Sigma, \delta, q_0, F)$.  Since
$\epsilon \not\in L(M)$, we know $q_0 \not\in F$.  
We create the machine $M' = (Q', \Sigma, \delta', q_0, F')$ as follows:
$$ \delta'(q,a) = \begin{cases}
	\{ \delta(q,a) \}, & \text{if $q \not\in F$}; \\
	\{ \delta(q_0, a), \ [\delta(q_0,a), \delta(q,a)] \} , & \text{if $q \in F$}.
	\end{cases}
$$
Writing $r = \delta(p,a)$, $s = \delta(q,a)$, $t = \delta(q_0, a)$, we also set
$$ \delta'([p,q],a) = \begin{cases}
	\{ [r,s] \}, & \text{if $p \not\in F$, $q \not\in F$}; \\
	\{ [r,s], [t, s] \}, & \text{if $p \in F$, $q \not\in F$}; \\
	\{ [r,s], [r, t] \}, & \text{if $p \not\in F$, $q \in F$}; \\
	\{ [r,s], [t, s], [r, t], [t, t] \}, & \text{if $p \in F$,
	$q \in F$}.
	\end{cases}
$$
Finally, we set $F' = F \times F$.
To see that the construction works, suppose that $x \in L^*$ has
two different factorizations
$$ x= y_1 y_2 \cdots y_j y_{j+1} \cdots y_k
= y_1 y_2 \cdots y_j z_{j+1} \cdots z_{\ell}$$
with $y_{j+1}$ a proper prefix of $z_{j+1}$.
Then an accepting path starts with singleton sets until the end
of $y_j$.  The next transition goes to a pair having first
element $\delta(q_0,a)$ with $a$ the first letter of $y_{j+1}$. 
Subsequent transitions eventually lead to a pair in $F \times F$.

On the other hand, if $x$ is accepted, then two different factorizations
are traced out by the accepting computation in each coordinate.
The factorizations are guaranteed to be different because of the
transition to $ [\delta(q_0,a), \delta(q,a)]$.
\qed
\end{proof}

\begin{remark}
There is a shorter and more transparent proof of this result, as follows.
Given a DFA for $L$, create an NFA $A$ for $L^*$ by adding 
$\epsilon$-transitions
from every final state back to the initial state, and then removing
the $\epsilon$-transitions using the familiar method 
(e.g., \cite[Theorem 2.2]{Hopcroft&Ullman:1979}).
Next, using the Boolean matrix interpretation of finite automata
(e.g., \cite{Zhang:1999} and \cite[\S 3.8]{Shallit:2009}), 
we can associate an adjacency
matrix $M_a$ with the transitions of $A$ on the letter $a$.
Then, on input $x = a_1 a_2 \cdots a_i$, a DFA can compute the matrix
$M_x = M_{a_1} M_{a_2} \cdots M_{a_i}$ using ordinary integer matrix
multiplication, with the proviso that any entry that is $2$ or more
is changed to $2$ after each matrix multiplication.   This can be done by
a DFA since the number of such matrices is at most $3^{n^2}$ where
$n$ is the number of states of $M$.  Then, accepting if and only if
the entry in the row and column corresponding to the initial state of
$A$ is $1$, we
get a DFA accepting exactly
those $x$ having unique factorization into elements of $L$.  While this
proof is much simpler, the state bound it provides is quite extravagant
compared to our previous proof.
\end{remark}

\begin{corollary}
Suppose $L$ is accepted by a DFA with $n$ states.  If $L$ is
not a code, then there exists a word $x \in L^*$ with at least two distinct
factorizations into elements of $L$, with $|x| < n^2 + n$.  
\end{corollary}

\begin{proof}
Our construction in the proof of Proposition~\ref{prop1}
gives an NFA $M'$ accepting all words with
at least two different factorizations, and it has $n^2 + n$ states.
If $M'$ accepts anything at all, it accepts a word of length
at most $n^2 + n - 1$.
\qed
\end{proof}

\begin{proposition}
For all $n \geq 2$, there exists an $O(n)$-state DFA
accepting a language $L$ that is not a code, such that
the shortest word in $L^*$ having two factorizations into elements of $L$
is of length $\Omega(n^2)$.
\end{proposition}

\begin{proof}
Consider the language $L_n = b (a^n)^* \ \cup \ (a^{n+1})^* b$.
It is easy to see that $L_n$ can be accepted by a DFA with
$2n+5$ states, but the shortest word in $L_n^*$ having two distinct
factorizations into elements of $L_n$
is $b \, a^{n(n+1)} \, b$, of length $n^2 + n + 2$.
\qed
\end{proof}

In fact, there are even examples of finite languages with the same property.

\begin{proposition}
For all $n \geq 2$, there exists an $O(n)$-state DFA
accepting a finite language $L$ that is not a code, such that
the shortest word in $L^*$ having two factorizations
is of length $\Omega(n^2)$.
\label{five}
\end{proposition}

\begin{proof}
Let $\Sigma = \{ b, a_1, a_2, \ldots, a_n \}$ be an alphabet
of size $n+1$, and let $L_n$ be the language of $2n$ words
$$ \{ a_1, a_n \} \ \cup \ \lbrace b^{i} a_{i+1} \ : \ 1 \leq i < n \rbrace
\ \cup \ \lbrace a_i b^i \ : \ 1 \leq i < n \rbrace $$
defined over $\Sigma$.

Then it is easy to see that $L_n$ can be accepted with a DFA
of $2n+2$ states, while the shortest word having two distinct
factorizations is
$$ a_1 b a_2 b^2 a_3 b^3 \cdots a_{n-1} b^{n-1} a_n,$$
which is of length $n(n+1)/2$.  
\qed
\end{proof}

\begin{remark}
The previous example can be recoded over a three-letter alphabet by
mapping each $a_i$ to the base-$2$ representation of $i$, 
padded, if necessary, to make it of length $\ell$, where 
$\ell = \lceil \log_2 n \rceil$.    With some reasonably obvious
reuse of states this can still be accepted by a DFA using
$O(n)$ states, and the shortest word with two distinct factorizations
is still of length $\Omega(n^2)$.
\end{remark}

\begin{theorem}
If $L$ is a CFL, then $\uf(L)$ need not be a CFL.
\end{theorem}

\begin{proof}
Let $L = {\tt PALSTAR}$, the set of all strings over the alphabet
$\Sigma = \{ 0, 1 \}$ that are the
concatenation of one or more even-length palindromes.  Clearly $L$ is
a CFL.  Then
$\uf(L) = {\tt PRIMEPALSTAR}$, which was proven in
\cite{Rampersad&Shallit&Wang:2011} to be non-context-free.
(Here ${\tt PRIMEPALSTAR}$ is the set of all elements of ${\tt PALSTAR}$
that cannot be written as the product of two or more elements of
${\tt PALSTAR}$.)
\qed
\end{proof}

\section{Semi-unique factorizations}

We now consider a variation on unique factorization.
We say that $x \in L^*$ has
{\it semi-unique factorization} if all factorizations of $x$ into
elements of $L$ consist of
the same number of factors.
More precisely, $x$ has semi-unique factorization if
whenever
$$x = y_1 y_2 \cdots y_m = z_1 z_2 \cdots z_n$$
for $y_1, y_2, \ldots, y_m, z_1, z_2,\ldots, z_n \in L$,
then $m = n$.  

Given a language $L$, we define $\su(L)$ to be the set of all elements
of $L^*$ having semi-unique factorization over $L$.

\begin{example}
Let $L = \{ a, ab, aab \}$.   Then $\su(L) = (ab)^* a^*$.
\end{example}

\begin{theorem}
If $L$ is regular, then $\su(L)$ is a CSL and a co-CFL. 
\end{theorem}

\begin{proof}
To see that $L$ is a co-CFL,
mimic the proof of Proposition~\ref{prop1}.
We use a stack to keep track of the difference between the
number of terms in the two guessed factorizations,
and another flag in the state to
say which, the ``top", or the ``bottom" state, has more terms (since
the stack can't hold negative counters).  We accept if we guess two
factorizations having different numbers of terms.

To see that $L$ is a CSL, note that $\su(L)$ is decidable in ${\tt DSPACE}(n)$.
(All we need to do is enumerate all the
possible factorizations; since no factorization is longer than the word
itself, we can list them all in linear space.)  
\qed
\end{proof}

\begin{corollary}
Given a regular language $L$, it is decidable if there exist elements
$x \in L^*$ lacking semi-unique factorization.
\end{corollary}

\begin{proof}
Given $L$, we can construct the PDA accepting $L^* - \su(L)$.  We 
convert this PDA to a CFG $G$ generating the same language
(e.g., \cite[Theorem 5.4]{Hopcroft&Ullman:1979}).
Finally, 
we use well-known techniques (e.g.,
\cite[Theorem 6.6]{Hopcroft&Ullman:1979})
to determine whether $L(G)$ is empty.
\qed
\end{proof}

\begin{theorem}
If $L$ is regular then $\su(L)$ need not be a CFL.
\end{theorem}

\begin{proof}
Let 
$$L = a 0^+ b + 1 + c (23)^+ + 23d + a + 0 + b 1^+ c (23)^+ +
a 0^+ b 1^+ c 2 + 32 + 3d.$$

Consider $\su(L)$ and intersect with the regular language
$a 0^+ b 1^+ c (23)^+ d$.

Then there are only three possible factorizations for a given word here.
They look like (using parens to indicate factors)

\bigskip

$  (a 0^i b) 1 \cdot 1 \cdot 1 \cdots 1 (c (23)^k) (23d)$, which has $j+3$
terms if $j$ is
  the number of 1's;

$  (a) 0 \cdot 0 \cdots 0 (b 1^j c (23)^k) (23d)$, which has $i+3$ terms if $i$ is
  the number of 0's; and

$  (a 0^i b 1^j c 2) (32) (32) \cdots (32) (3d)$, which has $k+2$ terms, if $k$ is
  the number of (32)'s.

\bigskip

  So if all three factorizations
have the same number of terms we must have $i = j = k-1$
  which gives us
$$  \{a 0^n b 1^n c (23)^{n-1} d : n \geq 1 \}$$
  which is not a CFL. 
\qed
\end{proof}

There are even examples where $L$ is finite.  For expository purposes,
we give an example over the $21$-letter alphabet
$$\Sigma = \{ 0,1,2,3,4,5,6,7,8,a,b,c,d,e,f,g,h,i,j,k,l \}. $$

\begin{theorem}
If $L$ is finite, then $\su(L)$ need not be a CFL.
\end{theorem}

\begin{proof}
Define
\begin{eqnarray*}
L_1 &=& \{ 0ab, cd, ab, cd127, efgh, efgh3, 4ijkl, ijkl, 5, 68 \} \\
L_2 &=& \{ 0abc, dabc, d1, 27e, fg, he, h34ij,klij,kl568 \} \\
L_3 &=& \{ 0a, bcda, bcd12, 7ef, ghef, gh34i, jk, li, jkl56, 8 \} 
\end{eqnarray*}
and set $L := L_1 \ \cup \ L_2 \ \cup \ L_3$.

Consider possible factorizations of words of the form
$$ 0 (abcd)^m 127 (efgh)^n 34 (ijkl)^p 568 $$
for some integers $m, n, p \geq 1$.
Any factorization of such a word into elements of $L$
must begin with either $0ab$, $0abc$, or $0a$.  There are
three cases to consider:

\bigskip

\noindent {\it Case 1:}
the first word is $0ab$.  Then the next word must begin with $c$, and
there are only two possible choices:  $cd$ and $cd127$.
If the next word is $cd$ then since no word begins with $1$ the only
choice is to pick a word starting with $a$, and there is only one:  $ab$.
After picking this, we are back in the same situation, and can only
choose between $cd$ followed by $ab$, or $cd127$.
Once $cd127$ is picked we must pick a word that begins with $e$.  However,
there are only two:  $efgh$ and $efgh3$.  If we pick $efgh$ we are left
in the same situation.  Once we pick $efgh3$ we must pick a word
starting with $4$, but there is only one:  $4ijkl$.  After this we can
either pick $5$ and then $68$, or we can pick $ijkl$ a number of times,
followed by $568$.

This gives the factorization
$$(0ab) ((cd)(ab))^{m-1} (cd127) (efgh)^{n-1} (efgh3) (4ijkl) (ijkl)^{p-1} (5) (68)$$
having $1 + 2(m-1) + 1 + (n-1) + 1 + 1 + (p-1) + 1 + 1 = 2m+n+p +2$ terms.

\bigskip

\noindent {\it Case 2:}
the first word is $0abc$.  Then the next word must begin with
$d$, and there are only two choices:  $dabc$ and $d1$.  If we pick
$dabc$ we are back in the same situation.  If we pick $d1$ then the 
next word must begin with $2$, but there is only one such word:  $27e$.  Then
the next word must begin with $f$, but there is only one:  $fg$.  
Then the next word must begin with $h$, but there are only two:
$he$ and $h34ij$.  If we pick $he$ we are back in the same situation.
Otherwise we must have a word beginning with $k$, but there are only
two:  $klij$ and $kl568$.  This gives the factorization

$$(0abc) (dabc)^{m-1} (d1) (27e) ((fg)(he))^{n-1} (fg) (h34ij) (klij)^{p-1} (kl568)$$
having
$1 + (m-1) + 2 + 2(n-1) + 1 + 1 + (p-1) + 1 = m+2n+p+2$ terms.

\bigskip

\noindent {\it Case 3:}
the first word is
$0a$.  Then only $bcda$ and $bcd12$ start with $b$, so we must choose $bcda$
over and over until we choose $bcd12$.  Only one word starts with $7$
so we must choose $7ef$.  Now we must choose $ghef$ again and again until
we choose $gh34i$.  We now choose $jk$ and $li$ alternately until $jkl56$.
Finally, we pick $8$.

This gives us a factorization
$$ (0a)(bcda)^{m-1} (bcd12) (7ef) (ghef)^{n-1} (gh34i) ((jk)(li))^{p-1} (jkl56) (8) $$
with $1 + (m-1) +  2 + (n-1) + 1 + 2(p-1) + 2 = m + n + 2p + 2$.

\bigskip

So for all these three factorizations to have the same number
of terms, we must have
$$2m + n + p + 2 = m + 2n + p + 2 = m + n + 2p + 2 .$$
Eliminating variables we get
that $m = n = p$.  So when we compute $\su(L)$ and intersect with the
regular language
$0 (abcd)^+ 127 (efgh)^+ 34 (ijkl)^+ 568$ we get
$$ \lbrace 0 (abcd)^n 127 (efgh)^n 34 (ijkl)^n 568  \ : \ n \geq 1 \rbrace,$$
which is clearly a non-CFL.
\qed
\end{proof}

\begin{remark}
The previous two examples can be recoded over a binary alphabet, by
mapping the $i$'th letter to the string $b a^i b$.  
\end{remark}

\section{Permutationally unique factorization}

In this section we consider yet another variation on unique factorization,
which are factorizations that are unique up to permutations of the
factors.

Formally, given a language $L$ we say $x \in L^*$ has
{\it permutationally unique factorization} if 
whenever $x = y_1 y_2 \cdots y_m = z_1 z_2 \cdots z_n$ for
$$y_1, y_2, \ldots, y_m, z_1, z_2, \ldots, z_n \in L,$$
then $m = n$ and there exists a permutation
$\sigma$ of $\lbrace 1, \ldots, n\rbrace$ such that
$y_i = z_{\sigma(i)}$ for $1 \leq i \leq n$.
In other words, we consider two factorizations that differ only in the
order of the factors to be the same.  We define $\ufp(L)$ to be the
set of $x$ having permutationally unique factorization.

\begin{example}
Consider $L = \{ a^3, a^4 \}$.  Then
$$\ufp(L) = \lbrace a^3, a^4, a^6, a^7, a^8, a^9, a^{10}, a^{11},
a^{13}, a^{14}, a^{17} \rbrace.$$
\end{example}

\begin{theorem}
If $L$ is finite then $\ufp(L)$ is a CSL and a co-CFL.
\label{jeff1}
\end{theorem}

\begin{proof}
The claim about CSL should be clear.

We sketch the construction of a PDA accepting $\overline{\ufp(L)}$.  
If a word is in $L^*$ but has two permutationally distinct factorizations,
then there has to be some factor appearing in the factorizations a different
number of times.  Our PDA nondeterministically guesses two different
factorizations and a factor $t \in L$ that appears a different number
of times in the factorizations, then verifies the factorizations and
checks the number.  It uses the stack to hold the absolute value of the
difference between the number of times $t$ appears in the first factorization
and the second.  It accepts if both factorizations end properly and the
stack is nonempty.
\qed
\end{proof}

\begin{theorem}
If $L$ is finite then $\ufp(L)$ need not be a CFL.
\label{bell}
\end{theorem}

\begin{proof}
Let $\Sigma = \{a, b, c\}$.
Define $L = \{A, B, S_1, S_2, T_1, T_2\} \subseteq \Sigma^+$ 
as follows:
$$
A = aa, \, B = aaa, \, S_1 = ab, \, S_2 = ac, \,T_1 = ba, \,T_2 = ca.
$$
Let $R = aa(ab)^+(ac)^+aa(ba)^+(ca)^+aaa$,
and consider words of the form
$$w := aa(ab)^r(ac)^saa(ba)^t(ca)^qaaa \in \textnormal{ufp}(L) \cap R$$
with $r, s, t, q \geq 1$ and the following two factorizations of $w$:
\begin{eqnarray}
AS_1^{r}S_2^{s}AT_1^tT_2^qB & = & aa\cdot (ab)^r\cdot (ac)^s \cdot aa \cdot (ba)^t \cdot (ca)^q \cdot aaa  \label{ufp1}\\
BT_1^{r}T_2^{s}S_1^{t}S_2^{q}AA & = & aaa \cdot (ba)^r\cdot (ca)^s\cdot (ab)^t \cdot (ac)^q \cdot aa \cdot aa \label{ufp2}
\end{eqnarray}
It is not difficult to see that $w$ must be of one of these two forms. Since $w$ has prefix $aaab$, it must start with either $AS_1$ or $BT_1$. If it starts with $AS_1 = aa\cdot ab$, the next factors must be $S_1^{r-1}$ to match $(ab)^{r}$, so we have $AS_1^{r}$. We then see $(ac)^s$, which can only match with $S_2^s$. Next, we see `$aaba$', thus we must choose $AT_1 = aa \cdot ba$. We then have $(ba)^{t-1}$, which can only match with $T_1^{t-1}$, and then $(ca)^q$, matching only with $T_2^q$. Finally the suffix is `$aaa$' which can only match with $B$ as required.

If $w$ starts with $BT_1 = aaa\cdot ba$, the next part is $(ba)^{r-1}$, which only matches with $T_1^{r-1}$. Then we see $(ca)^s$, so we must use factors $T_2^{s}$. We then see $(ab)^t$ and $(ac)^q$, matching with $S_1^t$ and $S_2^q$ respectively. Finally we have `$aaaa$' matching only with $AA$ as required.

If $r=t$ and $s=q$, then the number of each factor $(A, B, S_1, S_2, T_1, T_2)$ in factorizations \eqref{ufp1} and \eqref{ufp2} is identical.
Therefore, $w$ always has more than one factorization 
(of type \eqref{ufp1} or \eqref{ufp2}); however,
that factorization is only non-permutationally
equivalent if $r\neq t$ or $s \neq q$. Therefore 
\begin{eqnarray*}
\ufp(L) \cap R & = & \{aa\cdot (ab)^r\cdot (ac)^s \cdot aa \cdot (ba)^t \cdot (ca)^q \cdot aaa \mid (r=t) \ \wedge\  (s=q)\} \\
& = & \{ AS_1^{r}S_2^{s}AT_1^rT_2^sB \ : \ r, s \geq 1 \},
\end{eqnarray*}
which is not a context-free language.
\qed
\end{proof}

\section{Subset-invariant factorization}

In this section we consider yet another variation on unique factorization.
We say a word $x \in L^*$ has {\it subset-invariant factorization} (into
elements of $L$) if there exists a subset $S \subseteq L$ with the
property that every factorization of $x$ into elements of $L$ uses
exactly the elements of $S$ --- no more, no less ---
although each element may be used a
different number of times.  More precisely, $x$ has subset-invariant
factorization if
there exists $S = S(x)$ such that whenever $x = y_1 y_2 \cdots y_m$ with 
$y_1, y_2, \ldots, y_m \in L$, then $S = \lbrace y_1, y_2, \ldots, y_m \rbrace$.
We let $\ufs(L)$ denote the set of those $x \in L^*$ having such
a factorization.

\begin{theorem}
If $L$ is finite then $\ufs(L)$ is regular.
\end{theorem}

\begin{proof}
The proof is similar to the proof of Theorem~\ref{jeff1} above.  
On input $x$ we nondeterministically attempt to construct two
different factorizations into elements of $L$, recording which elements of
$L$ we have seen so far.  We accept if we are successful in constructing
two different factorizations (which will be different if and only if
some element was chosen in one factorization but not the other).  This
NFA accepts $L^* - \ufs(L)$.   So if $L$ is finite, it follows that
$\ufs(L)$ is regular.

In more detail, here is the construction.  States of our NFA are
$6$-tuples of the form $[w_1, s_1, v_1, w_2,  s_2,  v_2]$ where
$w_1, w_2$ are the words of $L$ we are currently trying to match;
$s_1, s_2$ are, respectively, the suffixes of $w_1$, $w_2$ we have
yet to see,
and $v_1, v_2$ are binary characteristic vectors
of length $|L|$, specifying  which elements
of $L$ have been seen in the factorization so far (including $w_1$ and
$w_2$, although technically they may not have been seen yet).
Letting $C(z)$ denote the vector with all $0$'s except a $1$ in the 
position corresponding to the word $z \in L$,
the initial states
are $[w, w, C(w), x, x, C(x)]$ for all words
$w, x \in L$.
The final
states are of the form $[w, \epsilon, v_1, x, \epsilon, v_2]$ where
$v_1 \not= v_2$.  Transitions on a letter $a$ look like
$\delta([w_1, as_1, v_1, w_2, as_2, v_2], a) = 
[w_1, s_1, v_1, w_2, s_2, v_2]$.
In addition there are $\epsilon$-transitions that 
update the corresponding vectors if $s_1$ or $s_2$ equals $\epsilon$,
and that ``reload'' the new $w_1$ and $w_2$ we are expecting to see:
\begin{eqnarray*}
\delta([w_1, \epsilon, v_1, w_2, s_2, v_2], \epsilon)
&=& \{ [w, w, v_1 \, \vee \, C(w), w_2, s_2, v_2] \ : \ w \in L \}  \\
\delta([w_1, s_1, v_1, w_2, \epsilon, v_2 ], \epsilon)
&=& \{ [w_1, s_1, v_1, w, w, v_2 \, \vee \, C(w)] \ : \ w \in L \}.
\end{eqnarray*}
\qed
\end{proof}

The preceding proof also shows that the shortest word failing to have
subset-invariant factorization is bounded polynomially:

\begin{corollary} 
Suppose $|L| = n$ and the length of the longest word of $L$ is $m$.
Then if some word of $L^*$ fails to have subset-invariant
factorization, there is a word with this property of length $\leq
2m^2n^2$.
\label{ufs}
\end{corollary}

\begin{proof}
Let $u \in L^+$ be a minimal length word such that $u \in L^+ -
\ufs(L)$. Consider the states of the NFA traversed in processing $u$.
Let $S_0 := [w, w, C(w), x, x, C(x)]$ be the initial state
and $S_F := [w_F, \epsilon, v_F, x_F, \epsilon, v'_F]$ the final state,
where $v_F \neq v'_F$. By definition, there must exist some $z \in L$
such that $v_F$ and $v'_F$ differ on $C(z)$, i.e., $v_F^T \cdot C(z) +
{v'_F}^T \cdot C(z) = 1$.

Initially the characteristic vectors have a single $1$, and once an
element is set to $1$ in a characteristic vector in the NFA, it is
never reset to $0$. Thus, there exists some $1 \leq k \leq |u|$ such
that $u = u_1 \cdots u_{k-1} \cdot u_{k} \cdot u_{k+1} \cdots u_{|u|}$
where $S_{k-1} = \delta(S_0, u_1\cdots u_{k-1})$ has a $0$ in the
characteristic vectors at position $z$, and $\delta(S_{k-1}, u_k)$ has
a $1$ in exactly one of the two characteristic vectors at position $z$.
We shall now prove that $|u_1 \cdots u_{k-1}|, |u_{k+1} \cdots u_{|u|}|
\leq m^2n^2$, which proves the result.

We prove the result for the word $v = u_1 \cdots u_{k-1}$; a similar
analysis holds for $u_{k+1} \cdots u_{|u|}$. 
Let 
$S_0, S_1, \ldots S_{k-1}$ be
the states of the
NFA visited as we process $v$.
We prove
that there does not exist $0 \leq i< j \leq k-1$ such that $S_i = [w_1,
s_1, v_1, w_2, s_2, v_2]$ and $S_j = [w_1, s_1, v'_1, w_2, s_2, v'_2]$.
We proceed by contradiction. Assume such an $i$ and $j$ exist. Then
$u_{i+1}\cdots u_{j}$ is such that $\delta(S_i, u_{i+1}\cdots u_{j}) =
S_j$. However, $\delta(S_i, u_{j+1}\cdots u_{k})$ and $\delta(S_j,
u_{j+1}\cdots u_{k})$ can only differ in their binary characteristic
vectors, since the transition function does not depend upon the 
characteristic vectors when we update the words $w_1, s_1, w_2, s_2$. Thus,
we can remove the factor $u_{i+1}\cdots u_{j}$ from $u$ and still reach a
final state of the form $S_{F_2} := [w_F, \epsilon, v_{F_2}, x_F,
\epsilon, v'_{F_2}]$, for which we still have that $v_{F_2} \neq
v'_{F_2}$, since they differ on element $z$ due to letter $u_k$.
Continuing this idea iteratively, the maximal number of states $k$ is
bounded by $m^2n^2$.  Doubling this bound gives the result.
\qed
\end{proof}

The next result shows that we can achieve a quadratic lower bound.

\begin{proposition}
There exist examples with $|L| = 2n$ and longest word of length
$n$ for which the shortest word of $L^*$ failing to have
subset-invariant factorization is of length $n(n+1)/2$.
\label{ufs2}
\end{proposition}

\begin{proof}
We just use the example of Proposition~\ref{five}.
\qed
\end{proof}

\begin{theorem}
If $L$ is regular then $\ufs(L)$ need not be a CFL.
\end{theorem}

\begin{proof}
We use a variation of the construction in the proof of Theorem~\ref{bell}.
Let $L = (ab)^+ (ac)^+ aa + (ba)^+(ca)^+ + aa + aaa $.
Then (using the notation in the proof of Theorem~\ref{bell}),
if $$w := aa(ab)^r(ac)^saa(ba)^t(ca)^qaaa \in \ufs(L) \cap R$$
with $r, s, t, q \geq 1$ then there are two different factorizations
of $w$:
\begin{eqnarray*}
w  & = & aa \cdot (ab)^r(ac)^s aa \cdot (ba)^t (ca)^q \cdot aaa \\
&=& aaa \cdot (ba)^r (ca)^s \cdot (ab)^t (ac)^q aa \cdot aa  
\end{eqnarray*}
which are subset-invariant if and only if $r = t$ and $s = q$.
So 
$$\ufs(L) \ \cap \ R = 
\{ aa (ab)^r (ac)^s aa (ba)^r (ca)^s aaa \ : \ r, s \geq 1 \},$$
which is not a CFL.
\qed
\end{proof}

\section{Acknowledgment}

The idea of considering semi-unique factorization was inspired by
a talk of Nasir Sohail at the University of Waterloo in April 2014.


\begin{thebibliography}{9}

\bibitem{Berstel&Perrin&Reutenauer:2010}
J. Berstel, D. Perrin, and C. Reutenauer.
\newblock {\it Codes and Automata}.
\newblock Encyclopedia of Mathematics and Its Applications, Vol.~129.
\newblock Cambridge University Press, 2010.

\bibitem{Hopcroft&Ullman:1979}
J. E. Hopcroft and J. D. Ullman.
\newblock {\it Introduction to Automata Theory, Languages, and Computation}.
\newblock Addison-Wesley, 1979.

\bibitem{Rampersad&Shallit&Wang:2011}
N. Rampersad, J. Shallit, and M.-w. Wang.
\newblock Inverse star, borders, and palstars.
\newblock {\it Info. Proc. Letters} {\bf 111} (2011), 420--422.

\bibitem{Shallit:2009}
J. Shallit.
\newblock {\it A Second Course in Formal Languages and Automata Theory}.
\newblock Cambridge University Press, 2009.

\bibitem{Zhang:1999}
G.-Q. Zhang.
\newblock Automata, Boolean matrices, and ultimate periodicity.
\newblock {\it Inf. Comput.} {\bf 152} (1999), 138--154.

\end{thebibliography}
\end{document}